\documentclass[conference]{IEEEtran}%
\usepackage{amsmath}
\usepackage{amsthm}
\usepackage{pgfplots}
\usepackage{amsfonts}
\usepackage{amssymb}
\usepackage{graphicx}
\usepackage{algorithmic}
\usepackage{algorithm}
\usepackage{todonotes}
\usepackage{mathtools}
\usepackage{manfnt}
\usepackage{balance}%
\providecommand{\U}[1]{\protect\rule{.1in}{.1in}}
\newtheorem{theorem}{Theorem}
\newtheorem{corollary}{Corollary}
\newtheorem{lemma}{Lemma}

\theoremstyle{definition}

\begin{document}

\title{Compressive Random Access Using A Common Overloaded Control Channel}
\author{Gerhard Wunder$^{1}$, Peter Jung$^{2}$ and Mohammed
Ramadan$^{3}$\\
$^{1}$Fraunhofer Heinrich-Hertz-Institut, Berlin
(\textit{gerhard.wunder@hhi.fraunhofer.de})\\
$^{2}$Technische Unversit{\"{a}}t Berlin
(\textit{peter.jung@tu-berlin.de})\\
$^{3}$Fraunhofer Heinrich-Hertz-Institut, Berlin
(\textit{mohammed.ramadan@hhi.fraunhofer.de})}
\maketitle
\let\thefootnote\relax\footnote{Gerhard Wunder is now also with the Freie
Unversit{\"{a}}t Berlin leading the new Heisenberg Group on Information and
Communication Theory. This work was carried out within
DFG grants WU 598/7-1 and WU 598/8-1 and JU-2795/3-1 (DFG Priority Program on Compressed Sensing),
and the 5GNOW project, supported by the European Commission within FP7 under grant 318555.
Peter Jung was also supported by DFG grant JU-2795/2.}

\maketitle

\begin{abstract}
We introduce a "one shot" random access procedure where users can send a
message without a priori synchronizing with the network. In this procedure a
common overloaded control channel is used to jointly detect sparse user
activity and sparse channel profiles. The detected information is subsequently
used to demodulate the data in dedicated frequency slots. We analyze the
system theoretically and provide a link between achievable
rates and standard compressing sensing estimates in terms of explicit
expressions and scaling laws. Finally, we support our
findings with simulations in an LTE-A-like setting allowing "one shot"
sparse random access of 100 users in 1ms.

\end{abstract}



\section{Introduction}

Sporadic traffic generating devices, e.g. machine-type communication (MTC),
are most of the time inactive but regularly access the Internet for
minor/incremental updates with no human interaction \cite{Wunder2014_COMMAG}.
Sporadic traffic will dramatically increase in the 5G market and, obviously,
cannot be handled with the bulky 4G random access procedures. Two major
challenges must be addressed:\ 1) unprecedent number of devices
asynchroneously access the network over a limited resource and 2)\ the same
resource carries control signalling and payload. Dimensioning the control
channels according to classical theory results in a severe waste of resources
which, even worse, does not scale towards the requirements of the IoT. On the
other hand, since typically user activity, channel profiles and message sizes
are compressible within a very large receive space, sparse signal processing
methodology is a natural framework to tackle the sporadic traffic.

\textbf{Preliminary Work}. The key findings of sparse signal processing are
that in an under-determined system undergoing noise the signal components can
be indeed identified if 1) the measurements are suitably composed and 2) the
signal space is sparse (or more generally "structured"), i.e. only a limited
number of elements in some given basis are non-zero. It has then soon been
recognized that this can be exploited for multiple access with sparse user
activity (see \cite{Wunder2015_ACCESS} for a recent overview). This was
extended to asynchronous fading channels \cite{Applebaum2012_PHYCOM} as well
as asynchronous multipath block fading channels with known/unknown (sparse)
multipath channel \cite{Dekorsy2013_ISWCS}. Notably, all these concepts are
fundamentally different from a classical "overloaded" CDMA channel
\emph{without sparsity }\cite{Wunder2006_TSP} where user activity and data
detection are separate steps.

\begin{figure*}[t]
\centering\includegraphics[width=1\linewidth]{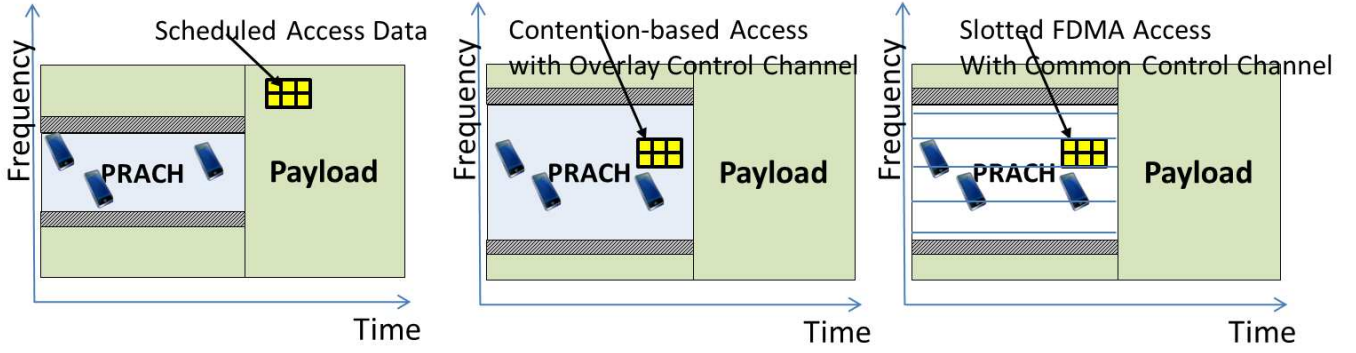}
\caption{\label{fig:prach:model} Random access concepts: a) standard procedure b) 5GNOW overlayed control channel \cite{Wunder2015_ACCESS,Wunder2014_ICC}
 c) 5GNOW common control channel (in this paper)}
\end{figure*}

Recently, massive MTC random access for 5G has fuelled the topic, particularly
within the EU projects METIS and 5GNOW
\cite{Popovski14,Wunder2014_GC,Wunder2014_ICC}. Ref. \cite{Popovski14}
investigates the interaction of advanced multiuser detection and a random
access scheme called \emph{coded slotted ALOHA}. However, the effect of
channel estimation and data detection errors (which is necessary due to
broadband) and error propagation in the intereference cancellation scheme is
crucial \cite{Popovski14} and must be carefully considered in follow-up work.
In \cite{Huang2013} new coding schemes and limits for MTC with random arrivals
in a \emph{slotted ALOHA} (using either time or frequency slots) have been
presented. While narrowband in nature and without advanced sparsity promoting
multiuser detection, it is argued that for such scenarios the effect of
channel estimation errors becomes negligible (however, implicitely, by using
arbitrary long preambles!). Consequently, recent concepts deal either with
data or channel estimation and an overall architecture which includes
identification, channel estimation, asynchronicity and data detection in
\textquotedblright one shot\textquotedblright\ is an open topic.

\textbf{Contributions}. In this paper, we propose a ALOHA (O)FDMA protocol
similar to \cite{Huang2013} where users select frequency slots of flexible
size without coordination. In such asynchronous scenario, it becomes very
inefficient to reserve control ressource for every slot since delay spread is
large, i.e. coherence bandwidth is small. We suggest an \emph{overloaded
common control channel} which is accessed by all active users at the same
time, see Fig. \ref{fig:prach:model}. Sparse signal processing will be used for joint (sparse) user activity
detection and (sparse) channel estimation. One control channel concept is that
data and control channels are superimposed so that control is spread over the
whole signal space and collected back within some (small) observation window
\cite{Wunder2014_ICC}. Another concept which is followed here is to fully
separate control and data which requires a careful investigation of actual
payload vs. control signaling ratio. In the following sections the details of
such sparsity aware random access scheme are outlined, analyzed and simulated.

\textbf{Notations}. $\lVert x\rVert_{\ell_{q}}=(\sum_{i}|x_{i}|^{q})^{1/q}$ is
the usual notion of $\ell_{q}$-norms and $\lVert x\rVert:=\lVert x\rVert
_{\ell_{2}}$, denote with $\text{supp}(x):=\{i\,:\,x_{i}:=\langle
e_{i},x\rangle\neq0\}$ the support of $x$ in a given fixed (here canonical)
basis $\{e_{i}\}_{i=1}^{n}$. The size of its support is denoted as $\lVert
x\rVert_{\ell_{0}}:=|\text{supp}(x)|$. $W$ is the (unitary) Fourier matrix
with elements $(W)_{kl}=n^{-\frac{1}{2}}e^{-i2\pi kl/n}$ for $k,l=0\dots n-1$,
hence, $W^{-1}=W^{\ast}$ where $W^{\ast}$ is the adjoint of $W$. We use here
also $\hat{x}=Wx$ to denote Fourier transforms and $\odot$ means point-wise
product. $I_{n}$ is the identity matrix in $\mathbb{C}^{n}$, diag$(x)$ is some
arbitrary diagonal matrix with $x\in\mathbb{C}^{n}$ on its diagonal.

\section{Compressive random access}

\subsection{"One-shot" transmitter}

Let $p_{u}\in\mathbb{C}^{n}$ be an pilot
(preamble) sequence from a given (random) set $\mathcal{P}\subset
\mathbb{C}^{n}$ and $x_{u}\in\mathbb{C}^{n}$ be an unknown (uncoded) data
sequence $x_{u}\in\mathcal{X}^{n}\subset\mathbb{C}^{n}$ both for the $u$-th
user with $u\in\{1,...,U\}$ and $U$ is the (fixed) maximum set of users in the
systems. Note that in our system $n$ is a very large number, e.g. 24k.
Due to the random zero-mean nature of $x_{u}$ we have $\frac{1}{n}E\lVert p_{u}+x_{u}\rVert^{2}=1$, i.e.
the total (normalized) transmit power is unity. Provided user $u$ is active we set:
\[
\alpha:=\frac{1}{n}\lVert p_{u}\rVert^{2}\quad\text{and}\quad\alpha^{\prime
}:=1-\alpha=\frac{1}{n}E\lVert x_{u}\rVert^{2}%
\]
Hence, the control signalling fraction of the power is $\alpha$.
If a user is not active then we set both $p_{u}=x_{u}=0$, i.e. either a user
is active and seeks to transmit data or it is inactive.

We will use a cyclic model, which is achieved with OFDM-like signaling and the
use of an appropriate cyclic prefix and restrict our model here to
time-invariant channels. Each vector $h_{u}\in\mathbb{C}^{T_{cp}}$ denotes the
sampled channel impulse response (CIR) where $T_{cp}$ is the length of the
cyclic prefix. We assume to have a priori support knowledge on each $h_{u}$:
(i) bounded support, i.e. $\text{supp}(h_{u})\subseteq\lbrack0,\dots
,T_{cp}-1]$ due to the cyclic prefix and (ii) sparsity, i.e. $\lVert
h_{u}\rVert_{l_{0}}\leq k_{1}$. Eventually, we assume that only $k_{2}$ users
out of $U$ in total are actually active. Define $k:=k_{1}k_{2}$.

In an OFDM system the FFT size $n$ is then chosen as $n\gg T_{cp}$. Let
$[h,0]\in\mathbb{C}^{n}$ denote the zero-padded CIR. With these assumptions
the received signal is then:
\begin{align*}
y  &  =\sum_{u=0}^{U-1}\text{circ}([h,0])(p_{u}+x_{u})+e\\
y_{\mathcal{B}}  &  =\Phi_{\mathcal{B}}y%
\end{align*}
Here, $\text{circ}([h,0])\in\mathbb{C}^{n}$ denotes the circulant matrix with
$[h,0]$ in its first column; $\Phi_{\mathcal{B}}$ denotes the overall
measurement matrix (to be specified later on) referring to a "frequency
observation window" $\mathcal{B}$. The number of subcarriers in $\mathcal{B}$
is $m:=|\mathcal{B}|$. The AWGN is denoted as $e\in\mathbb{C}^{n}$ with
$E(ee^{\ast})=\sigma^{2}I_{n}$. For circular convolutions we have
circ$([h,0])p=\sqrt{n}\cdot W^{\ast}(\hat{h}\odot\hat{p})$ so that:
\begin{align*}
y  &  =\sum_{u=1}^{U}W^{\ast}\left[  (\sqrt{n}\hat{h}_{u}\odot(\hat{p}%
_{u}+\hat{x}_{u})\right]  )+\hat{e}\\
y_{\mathcal{B}}  &  =\Phi_{\mathcal{B}}y
\end{align*}
where $e$ and $\hat{e}$ are statistically equivalent.

For the users' data the entire bandwidth $\mathcal{B}^{C}$ is dived into $B$
frequency slots. A standard assumption is that users' arrivals are modeled as
an Poisson process with rate $\lambda$ which is even true if retransmissions
are incorporated \cite{Huang2013}. Each user selects a slot in a ALOHA (O)FDMA
fashion. Clearly, by the single interval contention period the users will
choose the same FDMA slot with some probabilty, called outage event.
Obviously, for a fixed rate requirement throughput maximization means
outage probability minimization. Define the ordered user rates as
$R_{1},R_{2},...,R_{U}$. Then, it is shown in \cite{Huang2013} that the
average throughput (for any user $k$) is given by:
\[
T\left(  \lambda,\mathcal{R}\right)  =\lambda \exp\left(  -\frac{\lambda}%
{B}\right)  \cdot\Pr\left(  R_{k}>\mathcal{R}\right) \cdot\mathcal{R}
\]
with a rate constraint $\mathcal{R}$. In the
following we assume that the probabilties $\Pr(R_{k}>\mathcal{R})$ are mainly
dependent on the receive powers (user position, slow fading effects) while the
fast fading effects are averaged out due to coding over subcarriers. Hence,
user rates are ergodic and are calculated as expectations over the fading
distributions. The relevant expressions under erroneous channel estimation
will be provided in this paper.

\subsection{Receiver operations}

All performance indicators depend on the number of subcarriers in
$\mathcal{B}$ (control) and $\mathcal{B}^{C}$ (data). The goal is the
limitation to a small observation window $\mathcal{B}$. Let $P_{\mathcal{B}%
}:\mathbb{C}^{n}\rightarrow\mathbb{C}^{m}$ be the corresponding projection
matrix, i.e. the submatrix of $I_{n}$ with rows in $\mathcal{B}$. For
identifying which preamble is in the system we can consider $\hat{y}$ and use
the frequencies in $\mathcal{B}$, i.e. $\Phi_{\mathcal{B}}=P_{\mathcal{B}}W$,
so that:
\[
y_{\mathcal{B}}:=P_{\mathcal{B}}\sum_{u=1}^{U}\left[  \sqrt{n}\hat{h}_{u}%
\odot(\hat{p}_{u}+\hat{x}_{u})\right]  +P_{\mathcal{B}}\hat{e}%
\]
Notably in \cite{Wunder2014_ICC} we have introduced randomized pointwise
multipliers $\xi\in\mathbb{C}^{n}$ in time domain instead, denoted by the
corresponding $n\times n$ diagonal matrix $M_{\xi}:=\text{diag}(\xi)$, which
is favorable for the sparse recovery. Hence, the $m\times n$ sampling matrix
$\Phi_{\mathcal{B}}=P_{\mathcal{B}}WM_{\xi}$ is considered which, however,
comes at the cost of performance loss due to superimposed pilot subcarriers.

For algorithmic solution, we can stack the users as:%
\begin{align*}
y  &  =\sum_{u=1}^{U}\text{circ}(h_{u})(p_{u}+x_{u})+e\\
&  =D(p)h+C(h)x+e
\end{align*}
where $D(p):=[$circ$(p_{1}),\dots,$circ$(p_{U})]\in\mathbb{C}^{n\times Un}$
and $C(h):=[$circ$([h_{1},0]),\dots,$circ$([h_{U},0])]\in\mathbb{C}^{n\times
Un}$ are the corresponding compound matrices, respectively $p=[p_{1}%
^{T}\ p_{2}^{T}\ ...p_{U}^{T}]^{T}$ und $h=[h_{1}^{T}\ h_{2}^{T}\ ...h_{U}%
^{T}]^{T}$ are the corresponding compound vectors. If we assume each
user-channel vector $h_{u}$ to be $k_{1}$-sparse and $k_{2}$ are active then
$h$ is $k$-sparse.

For joint user activity detection and channel estimation exploiting the
sparsity we can use the standard basis pursuit denoising (BPDN) approach:%
\begin{equation}
\hat{\hbar}=\arg\min_{h}\lVert h\rVert_{\ell_{1}}\;\text{s.t. }\lVert
\Phi_{\mathcal{B}}\,D(p)h-y\rVert_{\ell_{2}}\leq\epsilon\label{eqn:bpdn}%
\end{equation}
Moreover, several greedy methods exists for sparse reconstruction. In
particular, for CoSAMP \cite{Needell08} explicit guarantees in reconstruction
performance are known and can be used instead of BPDN. After running the
algorithm in eqn. (\ref{eqn:bpdn}) the decision variables $\lVert\hbar
_{u}\rVert_{\ell_{2}}^{2}\;\forall u,$ are formed, indicating that if
$\lVert\hbar_{u}\rVert_{\ell_{2}}^{2}>\xi$ where $\xi>0$ is some predefined
threshold the user is considered active and its corresponding data is
detected. In \cite{Wunder2014_ICC} the correponding pilot signal is subtracted
from the recieved signal by interference cancellation. Here, we assume full
separation of data and control so that supp$(p_{u})\subseteq\mathcal{B}%
\;\forall u$. Denote the error of this operation as $d:=\hat{\hbar}-h$. Hence,
the received signal is given by:
\begin{equation}
\hat{y}=(\sqrt{n}\hat{\hbar}+\hat{d})\odot\hat{x}+\hat{e} \label{eqn:subc}%
\end{equation}
which a set of parallel channels each with power $E(|\hat{x}_{k}%
|^{2})=1-\alpha$, $|\hat{p}|^{2}=\alpha$ and $E(|\hat{e}_{k}|^{2})=\sigma^{2}$.

\section{Performance analysis}

It is possible to find the scaling of rates if one makes the following
assumptions: (i) all users employ independent Gaussian codebooks (ii) if a
user is not detected the corresponding data is discarded (iii) if
\emph{restricted isometry property} (RIP) is not satisfied (see the discussion
in following subsection), then the data of all is fully discarded
(iv) if the actual noise vector is larger than the
estimated noise vector then the data of all is discarded. In our analysis we
will work with outage probabilities for (i) and (ii). Refined estimates which
include (iii) and (iv) require more assumptions on the sampling model, noise
distribution and recovery procedure and will therefore appear in a separate work.

\subsection{RIP and performance guarantees}

Let $\Sigma_{k}:=\{x\in\mathbb{C}^{n}\,:\,\lVert x\rVert_{\ell_{0}}\leq k\}$
denotes the $k$-sparse vectors.
A matrix $\Phi$ is called $k$-RIP if exists $0\leq\delta_{k}<1$ such that
$|\lVert\Phi x\rVert^{2}-\lVert x\rVert^{2}|\leq
\delta_{k}\lVert x\rVert^{2}$
for all $x\in\Sigma_{k}$. There is a well-known result
\cite{candes:rip2008} on BPDN and we call this as the $Q_{1}$-estimator for
$x$ given $y$, i.e. $\tilde{x}=Q_{1}(y)$: If $\Phi_{\mathcal{B}}$ is $2k$-RIP
with $\delta_{2k}<\sqrt{2}-1$ and $\lVert e\rVert_{\ell_{2}}\leq\epsilon$
then:
\[
\lVert Q_{1}(\Phi x+e)-x\rVert_{\ell_{2}}\leq c_{1}\epsilon%
\]
with $c_{1}=4\sqrt{1+\delta_{2k}}/1-(1+\sqrt{2})\delta_{2k}$ (in particular,
for $\delta_{2k}=.2$ this gives $c_{1}=8.5$). It is known that $\delta
_{2k}\leq\sqrt{2}-1$ is a necessary condition. In \cite{Foucart2010} the bound
has been improved to $\delta_{2k}\leq3/(4+\sqrt{6})\approx0.4652$. Similar
bound exists for CoSAMP \cite{Needell08}, i.e. $\delta_{4k}\leq\sqrt
{2/5+\sqrt{73}}\approx0.3843$ \cite{Foucart2012}.

Up this point we have mentioned uniform reconstruction guarantees (for any
$x$) for a given matrix $\Phi$ and these are related to its RIP-constant
$\delta_{2k}$. It is still difficult to constructively design measurements
matrices with sufficiently small RIP constants. In this paper we use the
results in \cite{Rudelson:2007}: For any unitary matrix $U$, the measurement
matrix $P_{\mathcal{B}}\cdot U$ with $\mathcal{B}$ chosen uniformly at random
with cardinality $m$ such that:
\[
m\geq c^{\prime}\delta^{-2}\mu^{2}k\log^{5}(n)%
\]
has RIP with probability $\geq1-c^{\prime\prime}n^{-1}$ and $\delta_{2k}%
\leq\delta$, where $\mu=\mu(U,\text{Id})$ is the incoherence between $U$ and
the identity (standard basis).

\subsection{User detection}
Let us first calculate the propability of not detecting an active user
$P_{md}(\xi)$ ("missed detection"), and falsely detecting an inactive user
$P_{fa}(\xi)$ ("false alarm"). Recall that for a given pilot power $\alpha$
the channel estimation error is $d=\hbar-h$ with $\hbar=Q_{1}(y/\sqrt{\alpha
})$.
\begin{theorem}
We have for a fixed sampling matrix $\Phi$ with RIP-constant $\delta_{2k}$:%
\begin{align*}
P_{md}(\xi)  &  \leq F(\xi)+ \frac{c_{r}(\xi)c_{1}(\delta_{2k})^{2}m\sigma
^{2}}{\alpha k_{2}}\\
P_{fa}(\xi)  &  \leq\frac{c_{1}(\delta_{2k})^{2}m}{\alpha\xi\sigma^{2}}%
\end{align*}
where $c_{r}(\xi)$ is defined below in eqn. \eqref{eqn:c_r}.
\end{theorem}

\begin{proof}
Since the system is symmetric we can consider any user $u$ and drop the
indices $(\cdot)_{u}$ in all user--specific variables. Furthermore, we also
drop $\hat{\cdot}$ used to denote Fourier transforms. Abbreviate the
\emph{dependent} random variables by $x:=\left\Vert h\right\Vert $ and
$y:=\left\Vert d\right\Vert $ and let $F(x)$ be the distribution of $x$. We
have:%
\[
\begin{split}
&  \Pr\{\lVert\hbar\rVert<\xi\}=\Pr\{\lVert\hbar-h+h\rVert<\xi\}\\
&  \leq\Pr\{|\lVert d\rVert-\lVert h\rVert<\xi\}\leq\Pr\{x-\xi<y\}\\
&  =\Pr\{x-\xi<y\}\cap\{x<\xi\}+\Pr\{x-\xi<y\}\cap\{x\geq\xi\}\\
&  \leq\Pr\{x<\xi\}+\Pr\{x-\xi<y\}\cap\{x\geq\xi\}\\
&  =\Pr\{x<\xi\}+\Pr\{y^{2}>(x-\xi)^{2}\}\cap\{x\geq\xi\}\\
&  \leq\int_{0}^{\xi}dF(x)+\int_{\xi}^{\infty}\frac{E(y^{2}|x)}{(x-\xi)^{2}}dF(x)\\
&  \leq F(\xi)+ \frac{c_{r}(\xi)E(\left\Vert d\right\Vert ^{2})}{k_{2}}%
\end{split}
\]
where in the last step we use Markov's inequality and:%
\begin{equation}
c_{r}(\xi):=\int_{\xi}^{\infty}\frac{dF(x)}{(x-\xi)^{2}} \label{eqn:c_r}%
\end{equation}
and the last step follows since the $\ell_{2}$-norm function is the sum of its
squared terms. Hence, its expectation can be calculated as the expectation of
the sum of partial user contributions. These partial user contributions depend
only on the marginal distributions which are equal for all active users from
which the result follows.

The term $E(\lVert d\rVert^{2})$ can be calculated as%
\begin{align*}
E(\lVert d\rVert^{2})  &  =E(\lVert Q_{1}(y/\sqrt{\alpha})-h\rVert^{2})\\
&  \leq\frac{c_{1}(\delta_{2k})^{2}}{\alpha}E(\lVert e\rVert^{2})\\
&  \leq\frac{c_{1}(\delta_{2k})^{2}m\sigma^{2}}{\alpha}%
\end{align*}
The false alarm probability is calculated in a similar manner.
\end{proof}

For the (ergodic) achievable rates $R\left(  \alpha\right):=E\,\log(1+|h|^2)$
per subcarrier with a random channel power $|h|^2$
for a given
pilot/data power split $\alpha$ we can show the following:

\begin{theorem}
Let the channel impulse response be $k$-sparse and use eqn. \eqref{eqn:bpdn}
as the channel estimate. The achievable rate $R(\alpha)$ per subcarrier for a
particular user is lower bounded by:
\begin{align*}
R\left(  \alpha\right)   &  \geq E_{h|\{\left\Vert h\right\Vert >\xi\}}\left[
\log(1+\left(  1-\alpha\right)  |h|^{2}\sigma^{-2})\right]  \left(
1-P_{md}\right) \\
&  -\log\left(  1+\frac{\left(  1-\alpha\right)  c_{1}(\delta_{2k})^{2}%
m}{\alpha n}\right)
\end{align*}
for a fixed sampling $\Phi$ obeying a RIP-constant $\delta_{2k}<\sqrt{2}-1$.
\label{thm:2}
\end{theorem}

To prove this theorem we need the following lemma:

\begin{lemma}
\label{lemma:caire_general} Suppose $h,\hbar$ are random variables and $\hbar$
is an estimate of $h$ (hence they are correlated). Denote the error as
$d=h-\hbar$ and its unbiased version $\mathfrak{d}=d-E(d|\hbar)$. Then we
have:
\[
E\log\left(  1+|h|^{2}\right)  \leq E\log\left(  1+|\hbar+E(d|\hbar
)|^{2}+E(|\mathfrak{d}|^{2}|\hbar)\right)
\]

\end{lemma}

\begin{proof}
We have:
\begin{align*}
&  E\log\left(  1+|\hbar+E(d|\hbar)|^{2}+E(|\mathfrak{d}|^{2}|\hbar)\right) \\
&  =E\log\left(  1+E(|\hbar+E(d|\hbar)|^{2})|\hbar)+E(|\mathfrak{d}|^{2}%
|\hbar)\right) \\
&  =E\log\left(  1+E(|\hbar+E(d|\hbar)+\mathfrak{d})|^{2}|\hbar)\right) \\
&  =E\log\left(  1+E(|h|^{2}|\hbar)\right) \\
&  \geq E\log\left(  1+|h|^{2}\right)
\end{align*}
where the last line is due to Jenssen's inequality.
\end{proof}

\begin{corollary}
Note that Lemma \ref{lemma:caire_general} recovers Lemma 2 in \cite[Theorem
1]{Caire2010} by the fact that if $\hbar$ is the minimum mean squared estimate
(MMSE) of $h$ then $E(d|\hbar)=0$ so that $d=\mathfrak{d}$ and $\hbar$ and $d$
are actually independent. Hence $E(|\mathfrak{d}|^{2}|\hbar)=E(|\mathfrak{d}%
|^{2})=\alpha$ and $E(|\hbar|^{2})=1-\alpha$ so that:%
\begin{align*}
&  E\log\left(  1+|h|^{2}\right) \\
&  \leq E\log\left(  1+|\hbar+E(d|\hbar)|^{2}+E(|\mathfrak{d}|^{2}%
|\hbar)\right) \\
&  =E\log\left(  1+|\hbar|^{2}+E(|\mathfrak{d}|^{2}|\hbar)\right) \\
&  =E\log\left(  1+\left(  1-\alpha\right)  |h|^{2}+\alpha\right)
\end{align*}
because $\hbar$ and $h$ have the same distribution.
\end{corollary}

\begin{proof}
[proof of Theorem \ref{thm:2}] We have for each subcarrier and user (recall
that we also dropped user indices $(\cdot)_{u}$ and the $\hat{\cdot}$-notation for
the Fourier transform):
\[
y=(\hbar+d)x+e
\]
where $E(|x|^{2})=1-\alpha$ and $E(|e|^{2})=\sigma^{2}$. We can assume that
$x$ is independent of the tuple $(\hbar,d)$ but $\hbar$ and $d$ are in fact
not (they were if we had MMSE estimation). Linking the achievable mutual
information $I$ to MSE estimation as \cite[Theorem 1]{Caire2010}%
\begin{align*}
I(x;y|\hbar) &  =H(x|\hbar)-H(x|\hbar,y)\\
&  \geq H(x|\hbar)-E\left(  \log[\pi eE(\left\vert x-\alpha_{1}y\right\vert
^{2}|\hbar)]\right)
\end{align*}
where real $\alpha$ is some free parameter. The optimal $\alpha_{1}y$ is the
MMSE of $x$ so that in general by the orthogonality principle:
\[
\bar{\alpha}_{1}=E(yx^{\ast}|\hbar)/E(|y|^{2}|\hbar)
\]
In the following we use $\bar{\alpha}_{1}$ for bounding the rate but without
explizit calculating the term. Before we proceed let us rewrite
\begin{align*}
y &  =(\hbar+E(d|\hbar)+d-E(d|\hbar))x+e\\
&  =(\hbar+E(d|\hbar)+\mathfrak{d}))x+e
\end{align*}
where
\[
\mathfrak{d}=d-E(d|\hbar)
\]
Hence, we have $E(\mathfrak{d}|\hbar)=0$ so that $\hbar+E(d|\hbar)$ and
$\mathfrak{d}$ are in fact orthogonal which is the desired goal. Using optimal
$\bar{\alpha}_{1}$ we have after some tedious algebra:%
\begin{align*}
&  R(\alpha)\\
&  \geq E\log\left(  \frac{E(|x|^{2})|\hbar+E(d|\hbar)|^{2}+E(|x|^{2}%
)E(|\mathfrak{d}|^{2}|\hbar)+\sigma^{2}}{E(|x|^{2})E(|\mathfrak{d}|^{2}%
|\hbar)+\sigma^{2}}\right)  \\
&  =E\log\left(  1+\frac{E(|x|^{2})|\hbar+E(d|\hbar)|^{2}}{E(|x|^{2}%
)E(|\mathfrak{d}|^{2}|\hbar)+\sigma^{2}}\right)  \\
&  =E\log\left(  1+\frac{(1-\alpha)|\hbar+E(d|\hbar)|^{2}}{(1-\alpha
)E(|\mathfrak{d}|^{2}|\hbar)+\sigma^{2}}\right)
\end{align*}
The capacity loss per subchannel in eqn. (\ref{eqn:subc}) can be upperbounded
by
\begin{align*}
&  \Delta R\left(  \alpha\right)  \\
&  \leq E\log\left(  1+\frac{(1-\alpha)h}{\sigma^{2}}\right)  \\ & -E\log\left(
1+\frac{(1-\alpha)|\hbar+E(d|\hbar)|^{2}}{(1-\alpha)E(|\mathfrak{d}|^{2}%
|\hbar)+\sigma^{2}}\right)  \\
&  =E\log\left(  1+\frac{(1-\alpha)h}{\sigma^{2}}\right)  \\
&  -E\log\left(  1+\frac{(1-\alpha)|\hbar+E(d|\hbar)|^{2}+(1-\alpha
)E(|\mathfrak{d}|^{2}|\hbar)}{\sigma^{2}}\right)  \\
&  +E\log\left(  1+\frac{(1-\alpha)E(|\mathfrak{d}|^{2}|\hbar)}{\sigma^{2}%
}\right)
\end{align*}
which by virtue of Lemma \ref{lemma:caire_general} gives:
\[
\Delta R\left(  \alpha\right)  \leq E\log\left(  1+\frac{(1-\alpha
)E(|\mathfrak{d}|^{2}|\hbar)}{\sigma^{2}}\right)
\]
Clearly, we have $E(|\mathfrak{d}|^{2}|\hbar)\leq E(|d|^{2}|\hbar)$ and get so
the final result using:
\[
E(|\mathfrak{d}|^{2}|\hbar)\leq2E(|d|^{2}|\hbar)\leq\frac{c_{1}(\delta
_{2k})^{2}m\sigma^{2}}{\alpha n}%
\]
We can finally incorporate the missed detection properly.
\end{proof}

We also have the following upper bound.

\begin{theorem}
Let the channel impulse response be $k$-sparse and use eqn. (\ref{eqn:bpdn})
as the channel estimate. The achievable rate per subcarrier is upper bounded
by:
\[
R\left(  \alpha\right)  \leq E_{h}\log\left(  1+\frac{\left(  1-\alpha\right)
h\sigma^{-2}}{1+\frac{c_{1}(\delta_{2k})^{2}m}{n\alpha}}\right)
\]
for a fixed sampling $\Phi$ obeying a RIP-constant $\delta_{2k}<\sqrt{2}-1$.
\end{theorem}

\begin{proof}
The prove follows from \cite{Lapidoth2002_IT}.
\end{proof}

\section{Simulations}

An LTE-A 4G frame consists of a number of subframes with 20MHz bandwidth; the
first subframe contains the RACH with one "big" OFDM symbol of 839 dimensions
located around the frequency center of the subframe. The FFT size is 24578=24k
corresponding to the 20MHz bandwidth whereby the remainder bandwidth outside
PRACH is used for scheduled transmission in LTE-A, so-called PUSCH. The prefix
of the OFDM symbol accommodates delays up to 100$\mu s$ (or 30km cell radius)
which equals 3000 dimensions. In the standard the RACH is responsible for user
aquisition by correlating the received signal with preambles from a given set.
Here, to mimic a 5G situation, we equip the transmitter with the capability of
sending information in "one shot", i.e., in addition to user aquisition,
channel estimation is performed and the data is detected. For this a fraction
of the PUSCH is reserved for data packets of users which are detected in the
PRACH. Please note the rather challenging scenario of only 839 subcarrier in
the measurement window versus almost 24k data payload subcarriers.

In our setting, a limited number of users is detected out of a maximum set
(here 10 out of 100). We assume that the delay spread is below 300 dimensions
of which only a set of 6 pathes are actually relevant. The pilot signalling is similar to
\cite{Wunder2014_ICC} but modifed to fit the data/pilot separation. Each active user sends
1000 bits in some predefined frequency slot. This is uniquely achieved by
mapping the sequences to a slot. Hence, in the classical Shannon setting 100
users x (300 pathes + 1000 bits) = 130k dimensions are needed while there are
only 24k available! The performance results are depicted in
Fig. \ref{fig:crach:ser} where we show show symbol error rates (SER) over the
pilot-to-data power ratio $\alpha$. Moreover, in Fig. \ref{fig:crach:pfd} we
depict false detection probability $P_{FD}$ (some user is detected while not
active) over missed detection probability $P_{MD}$ (user is active while not
detected). We observe that, although the algorithms might not yet capture the
full potential of this idea, reasonable detection performance can be achieved
by varying $\alpha$. In the 4G LTE-A standard a minimum $P_{FD}=10^{-3}$ is
required for any number of receive antennas, for all frame structures and for
any channel bandwidth. For certain SNRs a minimum $P_{MD}=10^{-2}$ is
required. It can be observed from the simulations that the requirements can be
achieved. Actually, compared to 4G LTE-A where the control signalling can be
up to 2000\% \cite{Wunder2014_COMMAG} of a single resource element the control
overhead is in the CS setting down to to 13\% (let alone the huge increase in
latency) in the best case.

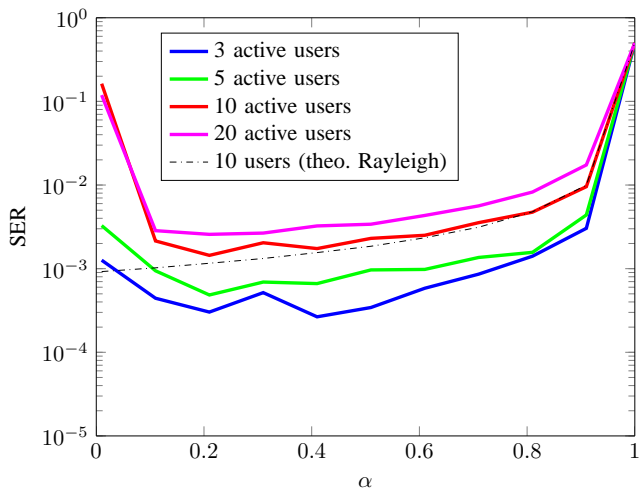
\begin{figure}[t]
\centering
\definecolor{mycolor1}{rgb}{1.00000,0.00000,1.00000}%
\begin{tikzpicture}[scale=.85]

\begin{axis}[%
width=0.951\linewidth,
height=0.739\linewidth,
at={(0\linewidth,0\linewidth)},
scale only axis,
separate axis lines,
every outer x axis line/.append style={black},
every x tick label/.append style={font=\color{black}},
xmin=0,
xmax=1,
xlabel={$\alpha$},
every outer y axis line/.append style={black},
every y tick label/.append style={font=\color{black}},
ymode=log,
ymin=1e-05,
ymax=1,
yminorticks=true,
ylabel={SER},
axis background/.style={fill=white},
legend style={at={(0.12,0.61)},anchor=south west,legend cell align=left,align=left,fill=none,draw=black}
]
\addplot [color=blue,solid,line width=1.5pt]
  table[row sep=crcr]{%
0.01	0.00126256630627449\\
0.11	0.00044303383500448\\
0.21	0.000302637861455043\\
0.31	0.000516845554618056\\
0.41	0.000265633891254439\\
0.51	0.000343181981699141\\
0.61	0.000583168180837078\\
0.71	0.000859716853403449\\
0.81	0.00140950297274493\\
0.91	0.00303708899920858\\
1	0.498041176932284\\
};
\addlegendentry{3 active users};

\addplot [color=green,solid,line width=1.5pt]
  table[row sep=crcr]{%
0.01	0.00327307147669224\\
0.11	0.000942970880835116\\
0.21	0.000484672703726233\\
0.31	0.000691702803995055\\
0.41	0.000662993250004905\\
0.51	0.00096603172890136\\
0.61	0.000979296254145164\\
0.71	0.0013612277534682\\
0.81	0.0015669431744599\\
0.91	0.00438431803465259\\
1	0.500378117457764\\
};
\addlegendentry{5 active users};

\addplot [color=red,solid,line width=1.5pt]
  table[row sep=crcr]{%
0.01	0.162856271652605\\
0.11	0.00214429368147261\\
0.21	0.0014421791844628\\
0.31	0.00204036534130426\\
0.41	0.00173815971063975\\
0.51	0.00230231791255966\\
0.61	0.00250683447904471\\
0.71	0.00354922852526997\\
0.81	0.0047368799150582\\
0.91	0.00955871212121212\\
1	0.500524364626189\\
};
\addlegendentry{10 active users};

\addplot [color=mycolor1,solid,line width=1.5pt]
  table[row sep=crcr]{%
0.01	0.118943622374555\\
0.11	0.00284773190034843\\
0.21	0.0025717851094892\\
0.31	0.00266344961125937\\
0.41	0.00324302167751394\\
0.51	0.00340034890779122\\
0.61	0.0043270061960958\\
0.71	0.00563269283362522\\
0.81	0.0082474562930406\\
0.91	0.0173883783471048\\
1	0.499858851822\\
};
\addlegendentry{20 active users};

\addplot [color=black,dashdotted]
  table[row sep=crcr]{%
0.01	0.000923194948055817\\
0.11	0.00102272521043062\\
0.21	0.00116047959743498\\
0.31	0.00132304074491651\\
0.41	0.00155682077477703\\
0.51	0.00185322203098892\\
0.61	0.00234067111650649\\
0.71	0.00313174194111054\\
0.81	0.00478103578180689\\
0.91	0.00995635573568282\\
1	0.499999999999995\\
};
\addlegendentry{10 users (theo. Rayleigh)};

\end{axis}
\end{tikzpicture}
\caption{Averaged BPSK SER in 5G \textquotedblright one-shot\textquotedblright\ random access in a $20$MHz LTE-A standard setting at (overall) SNR=$20$dB. $m=839$ out of
$n=24576$ dimensions are used for CS and sparsity of the channel is $k=6$. The
total number of users is 100, out of which 3-20 are active. The control
overhead is below $13$\%}%
\label{fig:crach:ser}%
\end{figure}

\begin{figure}[t]
\centering
\begin{tikzpicture}[scale=.85]

\begin{axis}[%
width=0.951\linewidth,
height=0.739\linewidth,
at={(0\linewidth,0\linewidth)},
scale only axis,
separate axis lines,
every outer x axis line/.append style={black},
every x tick label/.append style={font=\color{black}},
xmode=log,
xmin=0.0001,
xmax=0.1,
xminorticks=true,
xlabel={$P_{MD}$},
every outer y axis line/.append style={black},
every y tick label/.append style={font=\color{black}},
ymode=log,
ymin=1e-05,
ymax=1,
yminorticks=true,
ylabel={$P_{FD}$},
axis background/.style={fill=white},
legend style={at={(0.03,0.97)},anchor=north west,legend cell align=left,align=left,draw=black}
]
\addplot [color=blue,solid,line width=1.5pt,mark size=2.0pt,mark=o,mark options={solid}]
  table[row sep=crcr]{%
0.002	0.530444444444444\\
0	0.441333333333333\\
0.012	0.338222222222222\\
0.002	0.266888888888889\\
0.004	0.208444444444444\\
0.00600000000000012	0.147111111111111\\
0.0160000000000002	0.111777777777778\\
0.014	0.0851111111111116\\
0.012	0.0622222222222227\\
0.0140000000000002	0.0444444444444448\\
0.0460000000000004	0.00311111111111151\\
};
\addlegendentry{$\alpha$= 0.3};

\addplot [color=green,solid,line width=1.5pt,mark size=2.0pt,mark=o,mark options={solid}]
  table[row sep=crcr]{%
0	0.397333333333333\\
0	0.226222222222222\\
0.002	0.107111111111111\\
0	0.0593333333333339\\
0.002	0.0360000000000005\\
0.002	0.016444444444445\\
0.0100000000000001	0.0102222222222226\\
0.01	0.00488888888888928\\
0.00800000000000001	0.00222222222222257\\
0.0140000000000002	0.000666666666666815\\
0.0400000000000005	0\\
};
\addlegendentry{$\alpha$= 0.5};

\addplot [color=red,solid,line width=1.5pt,mark size=2.0pt,mark=o,mark options={solid}]
  table[row sep=crcr]{%
0	0.173777777777778\\
0	0.0411111111111118\\
0	0.00688888888888928\\
0	0.00177777777777788\\
0.002	0.000222222222222235\\
0.004	0.000222222222222235\\
0.00800000000000012	0\\
0.00800000000000001	0\\
0.012	0\\
0.0160000000000002	0\\
0.0360000000000003	0\\
};
\addlegendentry{$\alpha$= 0.7};


\draw[solid, line width=1.0pt] (axis cs:0.0001,1e-05) rectangle (axis cs:0.0101,0.00101);
\node[right, align=left, text=black] at (axis cs:0.0002,0.0001) {LTE standard};
\end{axis}
\end{tikzpicture}%
\caption{$P_{FD}$, $P_{MD}$ over $\xi$ and fixed $\alpha$ (such SER$< 10^{-3}$) for the 5G \textquotedblright one-shot\textquotedblright\ random access}
\label{fig:crach:pfd}%
\end{figure}
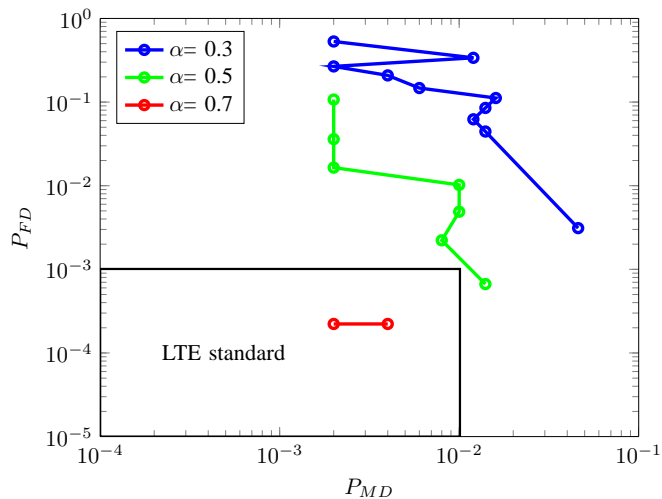

\section{Conclusions}

In this paper, we provided ideas how to enable random access for
many devices in a massive machine-type scenario. In the conceptional approach
as well as the actual algorithms sparsity of user activity and channel impulse
responses plays an a pivotal role. We showed that using such framework efficient
"one shot" random access is possible where users can send a
message without a priori synchronizing with the network. Key is a common
overloaded control channel which is used to jointly detect sparse user
activity and sparse channel profiles. Such common control channel stands in clear
contrast to dedicated control signalling per ressource block, and is thus
more efficent particularly for small ressoure blocks. Since each user
also has channel state information for all subcarriers, there are additional
degrees of freedom to place the ressource blocks.
We analyzed the system theoretically and provided a link between achievable
rates and standard compressing sensing estimates in terms of explicit
expressions and scaling laws. Finally, we supported our
findings with simulations in an LTE-A-like setting allowing "one shot"
sparse random access of 100 users in 1ms with good performance.

\section{Acknowledgements}

This work was carried out within DFG grants WU 598/7-1 and WU 598/8-1 (DFG Priority Program on Compressed Sensing),
and the 5GNOW project, supported by the European Commission within FP7 under grant 318555.
Peter Jung was supported by DFG grant JU-2795/2. We would also like to thank the reviewers for
their valuable comments.

\bibliographystyle{IEEEtran}
\bibliography{globecom15}

\end{document}